\documentclass{article}

\author{}

\usepackage{amsmath}
\usepackage{amsthm}
\usepackage{amsfonts}
\usepackage{amssymb}
\usepackage{latexsym}
\usepackage{setspace}
\usepackage[dvips]{graphicx}
\usepackage{epsfig}
\usepackage{color}

\usepackage[top=1.5in, bottom=1.5in, left=1in, right=1in]{geometry}

\newtheorem{theorem}{Theorem}[section]
\newtheorem{corollary}[theorem]{Corollary}

\newtheorem{proposition}{Proposition}[section]

\newtheorem{example}{Example}[section]

\newtheorem{remark}{Remark}[section]

\begin{document}
\large
\title{A systemic shock model for too big to fail financial institutions}
\author{Sabrina Mulinacci\footnote{University of Bologna, Department of Statistics, Via delle Belle Arti 41, 40126 Bologna, Italy. E-mail:
sabrina.mulinacci@unibo.it}}
\date{}
\maketitle
\begin{abstract}
In this paper we study the distributional properties of a vector of lifetimes in which each lifetime is modeled as the first arrival time between an idiosyncratic
shock and a common systemic shock. Despite unlike the classical multidimensional Marshall-Olkin model here only a unique common shock affecting all the lifetimes is assumed,
some dependence is allowed
between each idiosyncratic shock arrival time and the systemic shock arrival time. 
The dependence structure of the resulting distribution is studied through the analysis of its singularity and its associated copula function.
Finally, the model is applied to the analysis of the systemic riskiness of those European banks classified as systemically important (SIFI).
\end{abstract}
\bigskip

{\bf JEL Classification:} C51; G32

\bigskip

{\bf Keywords}: Marshall-Olkin distribution; Kendall's function;
Kendall's tau; systemic risk\bigskip

\section{Introduction}
In this paper we consider a particular generalization of the multidimensional Marshall-Olkin distribution (Marshall and Olkin, 1967) in the specific case in which, apart from the 
idiosyncratic ones, 
only one common shock is considered whose occurrence causes the simultaneous end of all lifetimes. More specifically, if $(X_0,X_1,\ldots , X_d)$ are positive
random variables that represent the arrival times of some shocks, then we consider, as resulting lifetimes the radom variables $T_1,\ldots , T_d$ defined as $T_j=\min (X_0, X_j)$,
$j=1,\ldots ,d$.

In the Marshall-Olkin model the underlying shocks arrival times are assumed to be independent and exponentially distributed. Many extensions exist in the literature
in order to consider marginal distributions different from the exponential one and to include some dependence among the underlying shocks arrival times,
even in the more general case where additional systemic shocks involving subsets of the lifetimes $T_1,\ldots , T_d$ are assumed. Among them, the scale-mixture of the Marshall-Olkin
distribution, introduced in Li (2009), is obtained by scaling, through a positive random variable, a random vector distributed according to the Marshall-Olkin distribution: this is equivalent
to assume that the underlying shocks arrival times have a dependence structure given by an Archimedean copula with a generator that is the Laplace transform of the mixing variable.
Scale-mixtures of the Marshall-Olkin distributions have also been considered in Mai et al. (2013) where, in the exchangeable case, a different construction is presented
involving L\'evy subordinators. 
On the other side, the approach of allowing for general marginal distributions in place of the exponential one, even preserving the independence, is studied in Li and Pellerey (2011) in the
bivariate case and extended to the multidimensional case in Lin and Li (2014): they call their distribution \emph{generalized Marshall-Olkin distribution}. Scale-mixtures of the generalized
Marshall-Olkin distribution are considered in Mulinacci (2015), with the aim, again, to introduce a specific Archimedean dependence (the generator is 
again given by the Laplace tranform of the mixing variable) among the underlying shocks arrival times:
the case of an underlying Archimedean dependence with a fully general generator is analyzed in Mulinacci (2017). 
The union of Marshall-Olkin and Archimedean dependence structures is also studied in Charpentier et al (2014).

The main drawback of all these extensions is that they 
assume an underlying exchangeable dependence. Aiming at considering an asymmetric underlying dependence, in Pinto and Kolev (2015), when $d=2$, 
the case in which $X_1$ and $X_2$ are dependent, while the external shock $X_0$ is independent of $(X_1,X_2)$ is studied.

The specific generalization of the Marshall-Olkin distribution presented in this paper is characterized by an asymmetric dependence
in the vector $(X_0,\ldots ,X_d)$ that goes in the opposite direction with respect to the one considered in Pinto and Kolev (2015): $X_1,\ldots ,X_d$ are assumed to be independent while a particular pairwise dependence is assumed between each $X_j$, $j=1,\ldots ,d$ and $X_0$. The 
pairwise dependence results from the assumption $X_0=\underset{j=0,1,\ldots ,d}\min Y_j$ where $Y_0,\ldots , Y_d$ are mutually independent while each $Y_j$ is correlated with $X_j$, $j=1,\ldots ,d$.

A possible branch of application of this model is in the reliability modeling of mechanical or electronic systems, and consequently, in the modeling of the resulting operational and actuarial risk. Consider for example $d$ working machines (or electronic components) $M_j$, $j=1,\ldots ,d$, all separately connected with a same machine 
$M_0$ so that if $M_0$ stops to working, immediately the same occurs for all the other machines. Assuming the classical Marshall-Olkin model, the failure of a single machine $M_j$, $j=1,\ldots ,d$ does not influence the failure of the machine $M_0$ or of the remaing $M_i$, $i=1,\ldots ,d, i\neq j$.
Conversely, in our model, the failure of one of the $M_j$, $j=1,\ldots ,d$ can influence the probability of failure of $M_0$, and, consequently, of 
the collapse of the whole system.
This is the case in which some electronic or mechanical desease in one of the $M_j$, $j=1,\ldots ,d$, being it connected with $M_0$, may worsen or interrupt the functioning 
status of $M_0$ to which can follow its failure.
\medskip

Another branch of application of this model is credit risk. There is a wide literature on applications of the Marshall-Olkin model and 
its generalizations to credit and actuarial risk (see Giesecke, 2003, Lindskog and McNeil, 2003,
Elouerkhaoui, 2007, Mai and Scherer, 2009, Baglioni and Cherubini, 2013, Bernhart et al., 2013 and Cherubini and Mulinacci, 2014).
Given the specific type of assumed dependence, the probabilistic model analyzed in this paper looks particularly suitable for the analysis of the joint lifetimes of the so called
Systemically Important Financial Institutions (SIFI) for which the default (or the proximity to it) of one of them, is directly correlated with the collapse of the whole system.
\bigskip

In this paper, we first discuss the survival distribution of the underlying vector of lifetimes $(X_0,\ldots ,X_d)$: we study the associated copula function and recover
expressions for the Kendall's function and Kendall's tau of the pairs $(X_0,X_j)$, $j=1,\ldots ,d$. Then we focus on the resulting joint
survival distribution of the lifetimes $(T_1,\ldots ,T_d)$: we analyze the probability of simultaneous end of all lifetimes (that is the singularity of the distribution)
and the dependence properties through the analysis of the pairwise Kendall's function and Kendall's tau. We do not make, in principle, any assumption on the maginal
distributions of the underlying shocks arrival times and on the underlying dependence structure: however, in order to obtain closed formulas, we restrict the analysis to particular 
classes of marginal distributions (that
include the exponential one as a particular case) and to Archimedean bivariate copulas.
Finally we present and discuss an application to the analysis of the systemic riskiness of European banks classified as SIFI by the 
Financial Stability Board: of course, the type of systemic risk modeled in this paper is very specific and the method is meant as an additional tool to analyze systemic risk
with respect to already existing ones.
\bigskip

The paper is organized as follows. In section \ref{model} we present and analyze the shocks arrival times model.
In section \ref{ol} we derive the distribution of the resulting, subjected to shocks, lifetimes which is, by construction, singular:
we compute the probability of the singularity and we analyze the dependece structure through the identification of the pairwise Kendall's function and Kendall's tau formulas.
In section \ref{sifi} we present an application to the analysis of the systemic riskiness of SIFI type European banks while section \ref{conclusion} concludes.

\section{The shocks arrival times model}\label{model}
Let us consider a general system whose components' lifetimes are denoted with $T_1,\ldots ,T_d$. We assume that each
lifetime is affected by an idiosyncratic shock causing the default of only that component and by a systemic shock causing the simultaneous default 
of all the components. The systemic shock arrival time is modeled as the first arrival time among $d+1$ shocks arrival times: one of them is
fully independent (as in the Marshall-Olkin model)
while each of the remaining ones is correlated with one of the idiosyncratic lifetimes components.
\bigskip

More formally, let $(\Omega,\mathcal F,\mathbb P)$ be a probability space and $\left ({\bf Y},{\bf X}\right )=(Y_0,Y_1,\ldots, Y_d,X_1,\ldots ,X_d)$,
be a $2d+1$-dimensional random vector with strictly positive elements:
we interpret the random variable $X_j$ ($j=1,\ldots ,d$) in ${\bf X}$ as the arrival time of a shock causing the dafault of only the $j$-th 
element in the system while $Y_j$ ($j=0,1,\ldots ,d$) in ${\bf Y}$
represents the arrival time of a shock causing the default of the whole system.

 We assume that the random variables in the sub-vector ${\bf X}$ are mutually independent as well as those in ${\bf Y}$ while some dependence is allowed in the pairs 
$(Y_j,X_j)$ for $j=1,\ldots ,d$.
All random variables $X_j$ in the subvector ${\bf X}$ have a 
survival distribution function denoted with
 $\bar F_{X_j}$, for $j=1,\ldots ,d$, strictly decreasing on $(0,+\infty)$. 
As for the random variables in the subvector ${\bf Y}$, in order to allow for the case of no shock arrival time correlated with 
the idiosyncratic lifetime component of some element in the system or for the case of no independent shock arrival time, we assume that the survival distribution function $\bar F_{Y_j}$ of each $Y_j$, for $j=0,\ldots ,d$, is strictly decreasing or identically equal to $1$ on $(0,+\infty )$:
however, we assume that there exists $j\in\{0,\ldots , d\}$ so that $\bar F_{Y_j}$ is not identically equal to $1$ on $(0,+\infty)$.
  
More precisely, the survival distribution function of $\left ({\bf Y},{\bf X}\right )$ is of type
$$\bar F_{\left ({\bf Y},{\bf X}\right )}(y_0,y_1,\ldots ,y_d,x_1,\ldots ,x_d)=\bar F_{Y_0}(y_0)\prod_{j=1}^d\hat C_j\left (\bar F_{Y_j}(y_j),\bar F_{X_j}(x_j)\right)$$
where $\left\{\hat C_j(u,v)\right\}_{j=1,\ldots ,d}$ is a family of bivariate copula functions: unlike $Y_0$, all other $Y_j$'s, for $j=1,\ldots ,d$,
are correlated to the hydiosincratic shock arrival times $X_j$ and $\hat C_j$ represents the survival dependence structure of the pair $(Y_j,X_j)$.\bigskip

\noindent
Starting from the above setup, we define the random variable
$$X_0=\min_{j=0,1,\ldots ,d}Y_j$$
that represents the first arrival time of a shock inducing the collapse of the whole system. As a consequence of the assumptions, its survival distribution function is of type
$$\bar F_{X_0}(x)=\prod_{j=0}^d\bar F_{Y_j}(x)$$
and $\bar
F_{X_0}$ is strictly decreasing on $(0,+\infty )$.
\medskip

Let us now consider the $d+1$-dimensional random vector ${\bf S}=(X_0,X_1,\ldots ,X_d)$ whose survival distribution function is
$$\begin{aligned}\bar F_{{\bf S}}(x_0,x_1,\ldots ,x_d)&=\mathbb P\left (Y_0>x_0,Y_1>x_0,\ldots ,Y_d >x_0, X_1>x_1,\ldots ,X_d>x_d\right )=\\
  &=\bar F_{Y_0}(x_0)\prod_{j=1}^d\hat C_j\left (\bar F_{Y_j}(x_0),\bar F_{X_j}(x_j)\right)
  \end{aligned}$$
Thanks to Sklar's theorem, the induced survival depedence structure is given by the survival copula
\begin{equation}\label{m1}\hat C(u_0,u_1,\ldots ,u_d)=\bar F_{Y_0}\circ \bar F_{X_0}^{-1}(u_0)\prod _{j=1}^d\hat C_j\left (\bar F_{Y_j}\circ \bar F_{X_0}^{-1}(u_0),u_j\right )\end{equation}
\begin{remark}\label{r1}
Since, for $j=0,\ldots ,d$, $g_j=\bar F_{Y_j}\circ \bar F_{X_0}^{-1}:[0,1]\rightarrow [0,1]$ is strictly increasing or identycally equal to $1$ and  $\prod _{j=0}^dg_j(v)=v$,
(\ref{m1}) represents a particular specification of the family of copulas introduced in Lemma 2.1 in Liebscher (2008).
\end{remark}
While the idyosincratic shocks arrival times $(X_1,\ldots ,X_d)$ are independent, by construction some dependence may exist only between each idyosincratic shock arrival time $X_j$, $j=1,\ldots ,d$, 
and the systemic one $X_0$.
More precisely, the survival distribution of each pair $(X_0,X_i)$ for $i=1,\ldots ,d$ is 
$$\begin{aligned}\bar F_{(X_0,X_i)}(x_0,x_i)
&=\hat C_i\left (\bar F_{Y_i}(x_0),\bar F_{X_i}(x_i)\right )\prod_{j=0,j\neq i}^d \bar F_{Y_j}(x_0)=\\
&=\hat C_i\left (\bar F_{Y_i}(x_0),\bar F_{X_i}(x_i)\right )\frac{\bar F_{X_0}(x_0)}{\bar F_{Y_i}(x_0)}  
  \end{aligned} $$
and the corresponding bivariate survival copulas are
\begin{equation}\label{biv}\hat C_{0,i}(u_0,u_i)=\hat C_i\left (\bar F_{Y_i}\circ \bar F_{X_0}^{-1}(u_0),u_i\right )\frac{u_0}{\bar F_{Y_i}\circ \bar F_{X_0}^{-1}(u_0)}.\end{equation}
Notice that while copulas $\hat C_i$ parametrize the dependence among the idyosincratic shock $X_i$ and
the arrival time $Y_i$ of a shock affecting the whole system, $\bar F_{Y_i}\circ \bar F_{X_0}^{-1}$ measures the contribution of the shock $i$ to the systemic shock arrival time $X_0$. \\
In order to analyze the dependence structure induced by (\ref{biv}) we compute the Kendall's function of a copula $\tilde C$ of type
\begin{equation}\label{m2}\tilde C(u,v)=C\left (g(u),v\right )\frac{u}{g(u)}\end{equation}
where $g:[0,1]\rightarrow [0,1]$ is strictly increasing.
\medskip

We remind that the Kendall's function of a bivariate copula $C(u,v)$ is defined as the cumulative distribution function of the random variable $C(U,V)$ where the random variables 
$U$ and $V$ are uniformly distributed on the interval $[0,1]$ and their joint distribution function is given by the considered copula $C(u,v)$.
More precisely the Kendall's function of a bivariate copula $C$ is a function $K:[0,1]\rightarrow [0,1]$ defined as
$$K_C(t)=\mathbb P\left (C(U,V)\leq t\right ),\, \text{ for }t\in [0,1]$$
(see Nelsen 2006, p. 127), where $\mathbb P$ is the probability induced by $C$. The relevance of this notion relies on the fact that it induces, through the corresponding one-dimensional stochastic ordering,
a partial ordering in the set of bivariate copulas: notice in particular that if $C_1(u,v)\leq C_2(u,v)$ for all $(u,v)\in [0,1]^2$, then $K_{C_1}(t)\geq K_{C_2}(t)$ for all $t\in [0,1]$
(see Nelsen, 2003, for more details).\\
Let us simplify the notation setting $\partial _1C(u,v))=\frac{\partial}{\partial u}C(u,v)$ for any copula $C(u,v)$.
\begin{proposition}\label{propo1}
Let $g:[0,1]\rightarrow [0,1]$ be strictly increasing and differentiable and the copula $C(u,v)$ be strictly increasing with respect to $v$ for any $u$.
Then the Kendall's function of a copula $\tilde C$ of type (\ref{m2}) is
$$K(t)=t-t\ln t+t\ln\left (g(t)\right )+\int_t^1\partial _1C(u,l_t(u))\frac{g^\prime (u)}{g(u)}udu$$
where $l_t(u)$ solves $\tilde C(u,l_t(u))=t$.
\end{proposition}
\begin{proof}
Since, for a given $u$, 
$C(u,v)$ is strictly increasing with respect to $v$, the inverse function $l_t(u)$ with respect to $v$ is  well defined for all $t\in (0,u]$ and satisfies
$\tilde C\left (g(u),l_t(u)\right )=\frac {g(u)}ut$.
Applying (6) in Genest and Rivest (2001) and after straightforward computations we have that 
$$\begin{aligned}
   K(t)&=t+\int_t^1\partial_1 \tilde C(u,l_t(u))du=\\
   &=t-t\ln\left (\frac t{g(t)}\right )+\int_t^1\partial _1C(u,l_t(u))\frac{g^\prime (u)}{g(u)}udu.
  \end{aligned}$$
\end{proof}

Let us now assume that the bivariate copula functions $\hat C_j$, for $j=1,\ldots ,d$ are of Archimedean type, with strict generator $\phi_j$: that is $\phi _j:[0,+\infty)\rightarrow (0,1]$ satisfies 
$\phi _j(0)=1$, 
$\underset{x\rightarrow +\infty}\lim \phi _j (x)=0$ and it is strictly decreasing and convex on $[0,+\infty)$ (see McNeal and Ne\v{s}lehov\'{a}, 2009).
Hence $\hat C_{0,i}$ in (\ref{biv}) takes the form
$$\hat C_{0,i}(u_0,u_i)=\phi_i\left (\phi_i^{-1}(\bar F_{Y_i}\circ \bar F_{X_0}^{-1}(u_0))+\phi_i^{-1}(u_i)\right )\frac{u_0}{\bar F_{Y_i}\circ \bar F_{X_0}^{-1}(u_0)}.$$
According to the general case, this copula is a particular specification of a copula of type
\begin{equation}\label{m3}\tilde C(u,v)=\phi\left (\phi ^{-1}\left (g(u)\right )+\phi ^{-1}(v)\right )\frac{u}{g(u)}.\end{equation}
The expression of the Kendall's function of a copula of this type can be immediately recovered from Proposition \ref{propo1}, taking into account that, now,
$l_t(u)=\phi\left (\phi^{-1}\left (\frac{g(u)}ut\right )-\phi^{-1}(g(u))\right )$. In fact it is a straightforward computation to verify that

\begin{corollary}
If $g:[0,1]\rightarrow [0,1]$ is strictly increasing and differentiable and $\tilde C(u,v)$ is a copula of type (\ref{m3}) with $\phi$ a strict Archimedean generator, 
then the Kendall's function of $\tilde C$ is
$$K(t)=t-t\ln t+t\ln\left (g(t)\right )+\int_t^1\frac{h\left (\frac{g(u)}ut\right )}{h\left (g(u)\right )}\frac{g^\prime (u)}{g(u)}udu$$
with $h(x)=\phi^\prime\circ\phi^{-1}(x)$.
\end{corollary}
The Kendall's function is strictly related to the widely used concordance measure known as Kendall's tau (see Section 5.1.1 in Nelsen, 2006). In fact,
the Kendall's tau $\tau$ can be obtained from the Kendall's function through
$$\tau=3-4\int_0^1K(t)dt$$
  \begin{example} \label{ex0} Let the function $g$ in (\ref{m3}) be of type $g(v)=v^{\theta }$, with $\theta\in (0,1]$ (this specific case was firstly introduced in Khoudraji, 1995): notice that this case is recovered in our model (see (\ref{biv}))
  when $\bar F_{Y_i}(x)=\bar F_{X_0}^{\theta}(x)$.  In this case we get
  $$K(t)=t-t\ln t+\theta t\ln t+\theta \int_t^1\frac{h\left (u^{\theta -1}t\right )}{h\left (u^{\theta}\right )}du$$
  and $$\tau=\theta-4\theta\int_0^1\int_t^1\frac{h\left (u^{\theta -1}t\right )}{h\left (u^{\theta}\right )}dudt.$$
In particular,
\begin{itemize}
 \item Clayton case, that is $\phi(x)=(1+x)^{-\frac 1\beta}$, with $\beta \geq 0$: since $h(y)=-\frac 1\beta y^{1+\beta}$, we have
\begin{equation}\label{fclay0}K(t)=t\left (1+\frac\theta\beta\right )-(1-\theta )t\ln t-\frac\theta\beta t^{1+\beta}\end{equation}
and  
\begin{equation}\label{clay0}\tau=\frac{\beta}{\beta +2}\theta=\tau^C_\beta\,\theta \end{equation}
where $\tau^C_\beta$ is the Kendall's tau of the Clayton copula with parameter $\beta$;
 \item Gumbel case, that is $\phi(x)=e^{-x^{\frac 1\beta}}$, with $\beta \geq 1$: since $h(y)=-\frac 1\beta y(-\ln y)^{1-\beta}$, we have
 $$K(t)=t-t\ln t\left [1-(\beta -1)\left (\frac \theta{1-\theta}\right) ^\beta 
 \int_{\frac\theta{1-\theta}}^{+\infty}\frac {1}{z^{\beta}(z+1)}dz\right ]$$
 and
 $$\tau=\left (1-\frac 1\beta\right )\left [\beta\left (\frac \theta{1-\theta}\right) ^\beta 
 \int_{\frac\theta{1-\theta}}^{+\infty}\frac {1}{z^{\beta}(z+1)}dz\right ]=\tau^G_\beta \left [\beta\left (\frac \theta{1-\theta}\right) ^\beta 
 \int_{\frac\theta{1-\theta}}^{+\infty}\frac {1}{z^{\beta}(z+1)}dz\right ]$$
where $\tau^G_\beta$ is the Kendall's tau of the Gumbel copula with parameter $\beta$.\end{itemize}
Notice that when $\theta =1$ we recover the Archimedean case: in our model this case corresponds to $\bar F_{Y_i}=\bar F_{X_0}$, that is the case in which
the only admissible fatal common shock is shock $i$.

\end{example}
\bigskip

\section{The lifetimes model}\label{ol}
In this section we study the joint distribution of the observed lifetimes $(T_1,T_2,\ldots ,T_d)$, each defined as the first arrival time between 
the corresponding idiosyncratic shock and the systemic one.
More precisely, for $j=1,\ldots ,d$, let
$$T_j=\min(X_j,X_0)$$
be the lifetime of the $j$-th element in the system. If we consider the random variables 
$$Z_j=Y_j\wedge X_j,$$
then, we can rewrite $T_j$ as
\begin{equation}\label{min_ind}T_j=\min \left (\min_{\underset{i\neq j}{i=0,\ldots, d}}Y_i, Z_j\right )\end{equation}
and each $T_j$ can also be modeled
as the first arrival time among $d$ independent shocks arrival times.
Since the survival distribution of $Z_j$ is
$$\bar F_{Z_j}(x)=\hat C_j\left (\bar F_{Y_j}(x),\bar F_{X_j}(x)\right ),\, x\geq 0,$$
it follows that the survival distribution of $T_j$ is
$$\bar F_{T_j}(x)=\hat C_j\left (\bar F_{Y_j}(x),\bar F_{X_j}(x)\right )\frac{\bar F_{X_0}(x)}{\bar F_{Y_j}(x)},\, x\geq 0.$$
More in general, the joint survival distribution function of 
${\bf T}=(T_1,\ldots ,T_d)$
can be easily recovered and it turns out to be given by
\begin{equation}\label{formula}\bar F_{{\bf T}}(t_1,\ldots ,t_d)=\bar F_{Y_0}\left (\max_{i=1,\ldots, d}t_i\right )\prod _{j=1}^d\hat C_j\left (\bar F_{Y_j}\left(\max_{i=1,\ldots, d}t_i\right),\bar F_{X_j}(t_j)\right)\end{equation}
for $(t_1,\ldots ,t_d)\in (0,+\infty)^d$.

The dependence structure implied by this survival distribution is the result of the joint contribution of the fact that the lifetimes can end 
simultaneously because of the occurrence of the systemic shock (which is the kind of dependence characteristic of the Marshall-Olkin distribution) and
of the fact that each element in the system can influence the occurrence of the systemic shock.
\begin{remark}
If $\hat C_j(u,v)=uv$, for all $j=1,\ldots ,d$, we get
$$\bar F_{{\bf T}}(t_1,\ldots ,t_d)=\bar F_{X_0}\left(\underset{i=1,\ldots d}\max t_i\right )\prod_{j=1}^d\bar F_{X_j}(t_j)$$
which is a particular specification of the generalized Marshall-Olkin distribution (see Li and Pellerey, 2011 and Lin and Li, 2014) with only one independent shock arrival time $X_0$.
\end{remark}
\begin{example}\label{ex1} 
Let us assume that the random variables $Y_0,Y_1,\ldots ,Y_d,Z_1,\ldots ,Z_d$ that generate the random variables $T_j$ 
(see (\ref{min_ind})) have survival distributions that belong to a same specific parametric family.
More precisely, we assume that
$$\bar F_{Y_j}(x)=G^{\gamma _j}(x),\, j=0,\ldots ,d\text{ and }\bar F_{Z_j}(x)=G^{\eta _j}(x),\, j=1,\ldots ,d$$
where $G$ is the survival distribution function of a strictly positive continuous random variable with support $(0,+\infty)$, 
$\gamma _j\geq 0$ (with at least one $j$ for which $\gamma_j >0$) and $\eta _j >0$.\\
Since $\bar F_{Z_j}(x)\leq \bar F_{Y_j}(x)$ we have that $\eta _j\geq \gamma _j$. We set $\lambda_j=\eta _j-\gamma _j$, for $j=1,\ldots ,d$ and $\lambda _0=\sum_{j=0}^d\gamma _j$.
It follows that
$$\bar F_{X_0}(x)=G^{\lambda _0}(x)\text{ and }\bar F_{T_j}(x)=G^{\lambda _0+\lambda _j}(x).$$
As a consequence (\ref{formula}) takes the form
$$\bar F_{{\bf T}}(t_1,\ldots ,t_d)=G^{\lambda _0}\left (\max_{i=1,\ldots, d}t_i\right )\prod _{j=1}^d
\hat C_j\left (G^{\gamma_j}\left(\max_{i=1,\ldots, d}t_i\right),\bar F_{X_j}(t_j)\right)$$
where $\bar F_{X_j}$ satisfies $\hat C_j\left (G^{\gamma_j}(x),\bar F_{X_j}(x)\right )=G^{\eta _j}(x)$
and the associated survival copula is
$$\hat C_{{\bf T}}(u_1,\ldots ,u_d)=\underset{i=1,\ldots d}\min u_i^{\frac{\gamma _0}{\lambda _0+\lambda _i}}\prod_{j=1}^d\hat C_j\left (
\underset{i=1,\ldots d}\min u_i^{\frac{\gamma _j}{\lambda _0+\lambda _i}},\bar F_{X_j}\left (G^{-1}\left (u_j^{\frac 1{\lambda _0+\lambda _j}}\right )\right )\right ).$$
\medskip

\noindent In the case in which $\hat C_j$ is Archimedean with strict generator $\phi_j$, we have
$$\bar F_{X_j}(x)=\phi_j\left (\phi_j^{-1}\left (G^{\eta_j}(x)\right )-\phi_j^{-1}\left (G^{\gamma_j}(x)\right )\right )$$
and 
$$\bar F_{{\bf T}}(t_1,\ldots ,t_d)=G^{\lambda _0}\left (\max_{i=1,\ldots, d}t_i\right )\prod _{j=1}^d
\phi_j\left (\phi_j^{-1}\left (G^{\gamma_j}\left(\max_{i=1,\ldots, d}t_i\right)\right)+\phi_j^{-1}\left (G^{\eta_j}(t_j)\right)-\phi_j^{-1}\left (G^{\gamma_j}(t_j)\right)
\right)$$
and
$$\hat C_{{\bf T}}(u_1,\ldots ,u_d)=\underset{i=1,\ldots d}\min u_i^{\frac{\gamma _0}{\lambda _0+\lambda _i}}\prod_{j=1}^d
\phi_j\left [\phi_j^{-1}\left (\underset{i=1,\ldots d}\min u_i^{\frac{\gamma _j}{\lambda _0+\lambda _i}}\right )
+\phi_j^{-1}\left (u_j^{\frac {\eta _j}{\lambda _0+\lambda _j}}\right )-\phi_j^{-1}\left (u_j^{\frac {\gamma _j}{\lambda _0+\lambda _j}}\right )\right ].$$
Let us set
$$\alpha _j=\frac{\lambda _0}{\lambda _0+\lambda _j},$$
which represents the ratio between the systemic shock intensity and the marginal one, and
$$\theta_j=\frac{\gamma_j}{\lambda _0},$$
which represents the percentage of contribution of the intensity of the shock correlated with each bank 
to the systemic shock intensity, for $j=1,\ldots ,d$, while $\theta _0$ is the percentage of contribution of some completely independent exogenous shock.
Then, we can rewrite the copula as
\begin{equation}\label{fc}\hat C_{{\bf T}}(u_1,\ldots ,u_d)=\underset{i=1,\ldots d}\min u_i^{\alpha _i\theta _0}\prod_{j=1}^d
\phi_j\left [\phi_j^{-1}\left (\underset{i=1,\ldots d}\min u_i^{\alpha _i\theta _j}\right )
+\phi_j^{-1}\left (u_j^{1-\alpha _j(1-\theta _j)}\right )-\phi_j^{-1}\left (u_j^{\alpha _j\theta _j}\right )\right ]\end{equation}
In particular 
\begin{itemize}
\item if $\phi _j$ is for all $j=1,\ldots ,d$ the Gumbel generator with parameter $\beta_j \geq 1$, then $\bar F_{X_j}(x)=G^{\left (\eta_j^{\beta_j}-\gamma_j^{\beta_j}\right )^{1/\beta_j}}(x)$ and
$$\hat C_{{\bf T}}(u_1,\ldots ,u_d)=\underset{i=1,\ldots d}\min u_i^{\alpha _i\theta _0}
\exp\left \{-\sum_{j=1}^d\left [\theta _j^{\beta _j}\underset{i=1,\ldots ,d}\max \{-\alpha _i\ln u_i\}^{\beta _j}+\sigma_j(-\ln u_j)^{\beta _j}
\right ]^{\frac 1{\beta _j}}\right \}$$
where $\sigma_j=(1-\alpha _j(1-\theta _j))^{\beta _j}-\alpha _j^{\beta _j}\theta _j^{\beta _j}$
\item if $\phi _j$ is for all $j=1,\ldots ,d$ the Clayton generator with parameter $\beta_j > 0$, then $\bar F_{X_j}(x)=\left (1+G^{-\eta_j\beta_j}(x)-G^{-\gamma_j\beta_j}(x)\right )^{-1/\beta_j}$ and 
\begin{equation}\label{cc1}\hat C_{{\bf T}}(u_1,\ldots ,u_d)=\underset{i=1,\ldots d}\min u_i^{\alpha _i\theta _0}\prod_{j=1}^d
\left [\left (\underset{i=1,\ldots ,d}\max u_i^{-\alpha _i}\right )^{\theta _j\beta _j}+u_j^{-(1-\alpha _j(1-\theta _j))\beta _j}
-u_j^{-\beta _j\alpha _j\theta _j}\right ]^{-\frac 1{\beta _j}}
\end{equation}
\end{itemize}
Notice that that, since $\bar F_{Y_j}(x)=\bar F_{X_0}^{\gamma _j/\lambda _0}(x)$, we recover the same framework considered in Example \ref{ex0}.
\medskip

\noindent From (\ref{fc}) the survival copula associated to $(T_i,T_k)$ is
$$\begin{aligned}&\hat C_{T_i,T_k}(u_i,u_k)=\\
   &=\left (\min (u_i^{\alpha _i},u_k^{\alpha _k})\right )^{1-\theta _i-\theta _k}
\prod_{j=i,k}\phi _j\left (\phi_j^{-1}\left (\left (\min (u_i^{\alpha _i},u_k^{\alpha _k})\right )^{\theta _j}\right )+
\phi_j^{-1}\left (u_j^{1-\alpha _j(1-\theta _j)}\right )-\phi_j^{-1}\left (u_j^{\alpha _j\theta _j}\right )\right)
  \end{aligned}
$$
from which, setting $\alpha _i=1$, we recover the survival copula associated to $(X_0,T_k)$
$$\hat C_{X_0,T_k}(u_i,u_k)=\frac{\min (u_i,u_k^{\alpha _k})}{\left (\min (u_i,u_k^{\alpha _k})\right )^{\theta_k}}
\phi _k\left (\phi_k^{-1}\left (\left (\min (u_i,u_k^{\alpha _k})\right )^{\theta _k}\right )+
\phi_k^{-1}\left (u_k^{1-\alpha _k(1-\theta _k)}\right )-\phi_k^{-1}\left (u_k^{\alpha _k\theta _k}\right )\right)
$$
 Notice that when $\alpha _k\approx0$ then
$$C_{X_0,T_k}(u_i,u_k)\approx \frac{u_i}{u_i^{\theta _k}}\phi_k\left (\phi_k^{-1}\left (u_i^{\theta _k}\right )+\phi_k^{-1}\left (u_k\right )\right )$$
which is of type (\ref{m3}) with $g(u)=u^\theta$ (see Example \ref{ex0}): the dependence structure of the observed lifetimes
with respect to the systemic shock arrival time essentially coincides with that between the idiosyncratic component of risk and the systemic one; in fact, being $\alpha_k\approx 0$,
$\lambda _j$ is large with respect to $\lambda _0$ and this corresponds to the case in which $X_j$ has a low survival distribution with respect to that of $X_0$.

\end{example}
\begin{remark}
In case of perfect dependence between each idiosyncratic shock and the corresponding systemic shock component, that is $\hat C_j(u,v)=\min (u,v)$ for $j=1,\ldots ,d$, we get
$$\bar F_{{\bf T}}(t_1,\ldots ,t_d)=\bar F_{Y_0}\left (\max_{i=1,\ldots, d}t_i\right )\prod _{j=1}^d\min\left (\bar F_{Y_j}\left(\max_{i=1,\ldots, d}t_i\right),\bar F_{X_j}(t_j)\right).$$
In the particular framework of Example \ref{ex1}, we have that, being $\min\left (G^{\gamma _j}(x),\bar F_{X_j}(x)\right )=G^{\eta _j}(x)$, if 
$\eta _j>\gamma _j$, then $\bar F_{X_j}(x)=G^{\eta _j}(x)$ and, if we assume $\eta_j >\gamma _j$ for all $j=1,\ldots ,d$, we can write
$$\bar F_{{\bf T}}(t_1,\ldots ,t_d)=G^{\gamma _0}\left (\max_{i=1,\ldots, d}t_i\right )\prod _{j=1}^d\min
\left (G^{\gamma _j}\left(\max_{i=1,\ldots, d}t_i\right),G^{\eta _j}(t_j)\right).$$
Conversely,if $\eta _j=\gamma _j$, then $\bar F_{X_j}(x)\geq G^{\gamma _j}(x)$ is not uniquely determined and if $\eta _j=\gamma _j$ for all $j=1,\ldots ,d$
$$\bar F_{{\bf T}}(t_1,\ldots ,t_d)=\bar F_{X_0}\left (\max_{i=1,\ldots, d}t_i\right ).$$

\end{remark}

\subsection{The probability of simultaneous default}
By construction, the distribution of ${\bf T}$ has a singularity generated by the occurrence of the simultaneous default 
of all the elements in the system, that is by the fact that 
the event $\{T_1=T_2=\cdots =T_d\}$ has positive probability. Both from a theoretical point of view as well as for applications, it is important to measure
the probability that the collapse of the whole system has to occur before a given time horizon.

\begin{proposition} If the random vector ${\bf T}$ has a survival distribution of type (\ref{formula}), then 
\begin{equation}\label{sing}\begin{aligned}\mathbb P\left (T_1=T_2=\cdots =T_d>t\right )&=-\int_t^{+\infty}\prod_{j=1}^d\bar F_{Z_j}(x)d\bar F_{Y_0}(x)+\\
&-\sum_{j=1}^d\int_t^{+\infty}\bar F_{Y_0}(x)\prod_{i\neq j}\bar F_{Z_i}(x)\,
\partial_1\hat C_j\left (\bar F_{Y_j}(x),\bar F_{X_j}(x)\right )d\bar F_{Y_j}(x)
\end{aligned}
\end{equation}
\end{proposition}
\begin{proof}
 $$\begin{aligned}
    &\mathbb P\left (T_1=T_2=\cdots =T_d>t\right )=
    \mathbb E\left [\mathbb P\left (\left .X_1>X_0,\ldots , X_d>X_0\right\vert X_0\right ){\bf 1}_{\{X_0>t\}}\right ]=\\
    &-\int_t^{+\infty}\prod_{j=1}^d\hat C_j\left (\bar F_{Y_j}(x),\bar F_{X_j}(x)\right )d\bar F_{Y_0}(x)+\\
&-\int _t^{+\infty}\bar F_{Y_0}(x)\sum_{j=1}^d\partial_1\hat C_j\left (\bar F_{Y_j}(x),\bar F_{X_j}(x)\right )\prod_{i\neq j}\hat C_i\left (\bar F_{Y_i}(x),\bar F_{X_i}(x)\right )d\bar F_{Y_j}(x).
   \end{aligned}
$$
\end{proof}
If each $\hat C_j$ is of Archimedean type with strict generator $\phi_j$, and $h_j=\phi_j^\prime\circ \phi _j^{-1}$, we have
$$\begin{aligned}\mathbb P\left (T_1=T_2=\cdots =T_d>t\right )&=-\int_t^{+\infty}\prod_{j=1}^d\bar F_{Z_j}(x)d\bar F_{Y_0}(x)+\\
&-\sum_{j=1}^d\int_t^{+\infty}\bar F_{Y_0}(x)\prod_{i\neq j}\bar F_{Z_i}(x)\frac{h_j\circ \bar F_{Z_j}(x)}
{h_j\circ \bar F_{Y_j}(x)}d\bar F_{Y_j}(x).
  \end{aligned}$$

\noindent
In particular, in the framework of Example \ref{ex1}, (\ref{sing}) takes the form
\begin{equation}\label{singsing}\mathbb P\left (T_1=T_2=\cdots =T_d>t\right )=\frac{\gamma_0}{\hat\lambda}G^{\hat\lambda}(t)+
\sum_{j=1}^d\gamma _j\int _0^{G(t)}y^{\hat \lambda -\lambda _j-1}
\frac{h_j\left (y^{\eta _j}\right )}
{h_j\left (y^{\gamma _j}\right )}dy
\end{equation}
where $\hat\lambda=\sum_{i=0}^d\lambda _j$. 
Hence
\begin{itemize}
 \item in the Clayton case (that is $\phi_j(x)=(1+x)^{-\frac 1{\beta _j}}$ with $\beta _j>0$ and $h_j(x)=-\frac 1{\beta _j}x^{1+\beta _j}$, for $j=1,\ldots ,d$), we have
  $$\mathbb P\left (T_1=T_2=\cdots =T_d>t\right )=\frac{\gamma_0}{\hat\lambda}G^{\hat\lambda}(t)+\sum_{j=1}^d\frac{\gamma _j}{\hat\lambda +\lambda _j\beta _j}G^{\hat\lambda +\lambda _j\beta _j}(t)$$
  and
 $$\begin{aligned}\mathbb P\left (T_1=T_2=\cdots =T_d\right )&=\frac{\gamma_0}{\hat\lambda}+\sum_{j=1}^d\frac{\gamma _j}{\hat\lambda +\lambda _j\beta _j}=\\
& =\frac{\theta_0}{\sum_{i=1}^d\alpha _i^{-1}-(d-1)}+\sum_{j=1}^d\frac{\theta_j}{\sum_{i=1}^d\alpha _i^{-1}-d+\beta_j(\alpha_j^{-1}-1)}
 ,\end{aligned}$$
\item in the Gumbel case (that is $\phi_j(x)=e^{-x^{\frac 1{\beta _j}}}$ with $\beta _j\geq 1$ and $h_j(x)=-\frac 1{\beta _j}x(-\ln x)^{1-\beta _j}$, for $j=1,\ldots ,d$), we have
  $$\mathbb P\left (T_1=T_2=\cdots =T_d>t\right )=\left (\frac{\gamma_0}{\hat\lambda}+\frac{1}{\hat\lambda}\sum_{j=1}^d\gamma _j\left (1+\frac{\lambda _j}{\gamma _j}\right )^{1-\beta _j}\right )G^{\hat\lambda}(t)$$
  and
 $$\begin{aligned}\mathbb P\left (T_1=T_2=\cdots =T_d\right )&=\frac{\gamma_0}{\hat\lambda}+\frac{1}{\hat\lambda}\sum_{j=1}^d\gamma _j\left (1+\frac{\lambda _j}{\gamma _j}\right )^{1-\beta _j}=\\
&=\frac{\theta_0}{\sum_{i=1}^d\alpha _i^{-1}-(d-1)}+
\sum_{j=1}^d\frac{\theta_j}{\sum_{i=1}^d\alpha _i^{-1}-(d-1)}\left (1+\frac{1-\alpha _j}{\alpha _j\theta_j}\right )^{1-\beta _j}.    
   \end{aligned}
.$$
\end{itemize}

\subsection{The Kendall's function and the Kendall's tau}
In this section we analyze the pairwise dependence structure of the random vector ${\bf T}$
through the study of the pairwise Kendall's function and the pairwise Kendall's tau.\medskip

\noindent In order to simplify the notation we set, for $i,k=1,\ldots ,d$, $i\neq k$,
$$P_{i,k}(x)=\frac{\bar F_{X_0}(x)}{\bar F_{Y_i}(x)\bar F_{Y_k}(x)}.$$
It can be easily checked that the survival distributions of the pairs $(T_i,T_k)$ are
$$\bar F_{i,k}(t_i,t_k)= P_{i,k}\left (\max(t_i, t_k)\right ) \prod_{j=i,k}\hat C_j\left (\bar F_{Y_j}(\max(t_i, t_k)),\bar F_{X_j}(t_j)\right ).
$$
\begin{proposition}
Let us assume that $\hat C_i$ and $\hat C_k$ are strictly increasing with respect to each argument. Then, for $t\in[0,1]$,
$$\begin{aligned}K_{i,k}(t)&=t-t\left (\ln\left (\frac{(\bar F_{Z_i}\cdot P_{i,k})\circ \bar F^{-1}_{T_i}(t)}{(\bar F_{Z_i}\cdot P_{i,k})(z_t)}\right )+
\ln\left (\frac{(\bar F_{Z_k}\cdot P_{i,k})\circ \bar F^{-1}_{T_k}(t)}{(\bar F_{Z_k}\cdot P_{i,k})(z_t)}\right )\right )+
\\
&-\int_{z_t}^{\bar F_{T_i}^{-1}(t)}\bar F_{Z_i} P_{ik}(x)\, \partial _1\hat C_k\left (\bar F_{Y_k}(x),\bar F_{X_k}(h_t(x))\right )
d\bar F_{Y_k}(x)+\\
&-\int_{z_t}^{\bar F_{T_k}^{-1}(t)}\bar F_{Z_k} P_{ik}(x)\, \partial _1\hat C_i\left (\bar F_{Y_i}(x),\bar F_{X_i}(g_t(x))\right )
d\bar F_{Y_i}(x)
\end{aligned}$$
where $z_t$ is the solution of $\bar F_{i,k}(z_t,z_t)=t$, $h_t(\cdot )$ solves $\bar F_{i,k}(x,h_t(x))=t$ for $z_t<x\leq \bar F_{T_i}^{-1}(t)$
and $g_t(\cdot )$ solves $\bar F_{i,k}(g_t(y),y)=t$ for $z_t<y\leq \bar F_{T_k}^{-1}(t)$.
\end{proposition}
\begin{proof}
Since $\bar F_{i,k}(x,x)=\bar F_{Z_i}(x)\bar F_{Z_k}(x)P_{i,k}(x)$ is strictly decreasing, given any $t\in[0,1]$, the solution of $\bar F_{i,k}(x,x)=t$, 
denoted with $z_t$, is well defined.
\\
If we restrict to $t_i>t_k$, then
\begin{equation}\label{restriction}\bar F_{i,k}(t_i,t_k)=\bar F_{Z_i}(t_i)\hat C_k\left (\bar F_{Y_k}(t_i),\bar F_{X_k}(t_k)\right )P_{i,k}(t_i)\end{equation}
which is strictly decreasing with respect to $t_k\in[0,t_i)$ for any given $t_i$. Hence, for $x\in (z_t,\bar F_{T_i}^{-1}(t)]$ and for any $t\in [0,1]$, the function $h_t$ satisfying 
$\bar F_{i,k}(x,h_t(x))=t$ is well defined.
\\By similar arguments, the function $g_t$ of the statement is also well defined.
\medskip

If we denote with $K_{i,k}$ the
Kendall's function associated to the pair $(T_i,T_k)$ and we rewrite it in terms of the survival joint distribution function, we get
$$\begin{aligned}
  K_{i,k}(t)
  &=\mathbb P(\bar F_{i,k}(T_i,T_k)\leq  t)=\\
&=\bar F_{T_i} (z_t)-\mathbb P((T_i,T_k)\in \mathcal D_1)+\bar F_{T_k} (z_t)-\mathbb P((T_i,T_k)\in \mathcal D_2)-t 
  \end{aligned}
$$
where
$$\mathcal D_1=\{(t_i,t_k): z_t<t_i\leq \bar F_{T_i}^{-1}(t), 0\leq t_k\leq h_t(t_i)\}$$
and 
$$\mathcal D_2=\{(t_i,t_k): z_t<t_k\leq \bar F_{T_k}^{-1}(t), 0\leq t_i\leq g_t(t_k)\}.$$
\medskip

Let us start computing $\mathbb P((T_i,T_k)\in \mathcal D_1)$. 
Since here (\ref{restriction}) holds, thanks to the definitions of $z_t$ and $h_t$, we have 
$$\begin{aligned}&\mathbb P\left ((T_i,T_k)\in\mathcal D_1\right )=
\int_{z_t}^{\bar F_{T_i}^{-1}(t)}\left (\mathbb P(T_k> h_t(x)\vert T_i=x)-1\right )d\bar F_{T_i}(x)=\\
&=\int_{z_t}^{\bar F_{T_i}^{-1}(t)}\mathbb P(T_k> h_t(x)\vert T_i=x)d\bar F_{T_i}(x)-t+\bar F_{T_i}(z_t) =\\
&=\int_{z_t}^{\bar F_{T_i}^{-1}(t)}t\cdot\frac {d(\bar F_{Z_i}\cdot P_{ik})(x)}{\bar F_{Z_i}\cdot P_{ik}(x)}+\\
&+\int_{z_t}^{\bar F_{T_i}^{-1}(t)}\bar F_{Z_i}\cdot P_{ik}(x)
\cdot \partial _1\hat C_k\left (\bar F_{Y_k}(x),\bar F_{X_k}(h_t(x))\right )d\bar F_{Y_k}(x)-t+\bar F_{T_i}(z_t)=\\
&=t\ln\left (\frac{(\bar F_{Z_i}\cdot P_{i,k})\circ \bar F^{-1}_{T_i}(t)}{(\bar F_{Z_i}\cdot P_{i,k})(z_t)}\right )+\\
&+\int_{z_t}^{\bar F_{T_i}^{-1}(t)}\bar F_{Z_i}P_{ik}(x)\cdot \partial _1\hat C_k\left (\bar F_{Y_k}(x),\bar F_{X_k}(h_t(x))\right )
d\bar F_{Y_k}(x)-t+\bar F_{T_i}(z_t).
\end{aligned}$$
Since $\mathbb P((T_i,T_k)\in \mathcal D_2)$ can be similarly computed, we get 
$$\begin{aligned}K_{i,k}(t)&=t-t\left (\ln\left (\frac{(\bar F_{Z_i}\cdot P_{i,k})\circ \bar F^{-1}_{T_i}(t)}{(\bar F_{Z_i}\cdot P_{i,k})(z_t)}\right )+
\ln\left (\frac{(\bar F_{Z_k}\cdot P_{i,k})\circ \bar F^{-1}_{T_k}(t)}{(\bar F_{Z_k}\cdot P_{i,k})(z_t)}\right )\right )+
\\
&-\int_{z_t}^{\bar F_{T_i}^{-1}(t)}\bar F_{Z_i}P_{ik}(x)\cdot \partial _1\hat C_k\left (\bar F_{Y_k}(x),\bar F_{X_k}(h_t(x))\right )
d\bar F_{Y_k}(x)+\\
&-\int_{z_t}^{\bar F_{T_k}^{-1}(t)}\bar F_{Z_k}P_{ik}(x)\cdot \partial _1\hat C_i\left (\bar F_{Y_i}(x),\bar F_{X_i}(g_t(x))\right )
d\bar F_{Y_i}(x).
\end{aligned}$$
\end{proof}

If we consider the case in which $\hat C_i$ and $\hat C_k$ are Archimedean copulas with strict generator $\phi_i$ and $\phi_k$, respectively, we have
$$\begin{aligned}K_{i,k}(t)&=t-t\left (\ln\left (\frac{(\bar F_{Z_i}\cdot P_{i,k})\circ \bar F^{-1}_{T_i}(t)}{(\bar F_{Z_i}\cdot P_{i,k})(z_t)}\right )+
\ln\left (\frac{(\bar F_{Z_k}\cdot P_{i,k})\circ \bar F^{-1}_{T_k}(t)}{(\bar F_{Z_k}\cdot P_{i,k})(z_t)}\right )\right )+
\\
&-\int_{z_t}^{\bar F_{T_i}^{-1}(t)}\bar F_{Z_i}\cdot P_{ik}(x)\frac{h_k\left (\frac t{\bar F_{Z_i}\cdot P_{i,k}(x)}\right )}
{h_k\circ \bar F_{Y_k}(x)}d\bar F_{Y_k}(x)+
\\
&-\int_{z_t}^{\bar F_{T_k}^{-1}(t)}\bar F_{Z_k}\cdot P_{ik}(x)\frac{h_i\left (\frac t{\bar F_{Z_k}\cdot P_{i,k}(x)}\right )}
{h_i\circ \bar F_{Y_i}(x)}d\bar F_{Y_i}(x)
\end{aligned}$$
with $\bar F_{Z_j}(x)=\phi_j\left (\phi_j^{-1}(\bar F_{Y_j}(x))+\phi_j^{-1}(\bar F_{X_j}(x))\right )$ and $h_j=\phi_j^\prime\circ\phi_j^{-1}$, for $j=i,k$.
\begin{example}\label{ex2}
Let us consider the same framework of Example \ref{ex1}. We have
$$z_t=G^{-1}\left (t^{1/(\lambda_0+\lambda_i+\lambda _k)}\right )$$
and 
$$\bar F_{T_j}^{-1}(t)=G^{-1}\left (t^{1/(\lambda_0+\lambda_j)}\right ),\, j=i,k.$$
Under the same notation of Example \ref{ex1}, we recover
$$\begin{aligned}K_{i,k}(t)&=t-t\ln t\left (
\frac{\alpha _i(1-\alpha _k)(1-\alpha _i\theta_ k)}{\alpha _i+\alpha _k-\alpha_i\alpha _k}
+\frac{\alpha _k(1-\alpha _i)(1-\alpha _k\theta_ i)}{\alpha _i+\alpha _k-\alpha_i\alpha _k}
\right )-\\
&-\theta_k\int_{t^{\frac{\alpha_i\alpha_k}{\alpha_i+\alpha _k-\alpha _i\alpha _k}}}^{t^{\alpha _i}}
y^{\frac{1-\alpha _i}{\alpha _i}}\frac{h_k\left (ty
^{-\frac{1-\theta_k\alpha _i}{\alpha _i}}
\right )}{h_k(y^{\theta_k})}
dy+\\
&-\theta_i\int_{t^{\frac{\alpha_i\alpha_k}{\alpha_i+\alpha _k-\alpha _i\alpha _k}}}^{t^{\alpha _k}}
y^{\frac{1-\alpha _k}{\alpha _k}}\frac{h_i\left (ty
^{-\frac{1-\theta_i\alpha _k}{\alpha _k}}
\right )}{h_k(y^{\theta_i})}
dy.
\end{aligned}$$
Let us consider Clayton and Gumbel copulas specific cases.
 \begin{enumerate}
  \item Clayton case ($\phi_j(x)=(1+x)^{\frac 1{\beta_j}}$, $\beta _j> 0$, $j=i,k$).
  
If we set $$\tau^{MO}_{ik}=\frac{\alpha_k\alpha_i}{\alpha_k+\alpha_i-\alpha_k\alpha_i}$$
which is the Kendall's tau of the Marshall-Olkin bivariate copula with parameters $\alpha_i$ and $\alpha _k$ and 
 $$\rho_{rs}=\frac{1-\alpha _s}{\alpha _s}\tau^{MO}_{rs}, r,s=i,j$$
 we get 
 $$\begin{aligned}K_{i,k}(t)
&=t\left (1+\frac{\theta_k}{\beta _k}\alpha _i+\frac{\theta_i}{\beta _i}\alpha _k\right )-t\ln t\left ((1-\theta_k\alpha _i)\rho_{ik}+
(1-\theta_i\alpha _k)\rho_{ki}\right )+\\
&-\frac{\theta_k}{\beta _k}\alpha _it^{\rho_{ik}\beta _k+1}-\frac{\theta_i}{\beta _i}\alpha _kt^{\rho_{ki}\beta _i+1}.
\end{aligned}
$$ 
Notice that the above Kendall's function can be decomposed as
$$K_{i,k}(t)=K_{ik}^0(t)+K_{0,k}^{(i)}+K_{0,i}^{(k)}-2K^I(t)$$
where
\begin{equation}\label{kfc0}K_{ik}^0(t)=t-\left (1-\tau^{MO}_{ik}\right )t\ln t,\end{equation}
\begin{equation}\label{kfc1}K_{0,k}^{(i)}=t\left (1+\frac{\theta_k\alpha _i}{\beta_k}\right )-(1-\theta_k\alpha _i\rho_{ik} )t\ln t-
\frac{\theta_k\alpha _i}{\beta_k} t^{1+\beta_k\rho_{ik}},
\end{equation}
\begin{equation}\label{kfc2}K_{0,i}^{(k)}=t\left (1+\frac{\theta_i\alpha _k}{\beta_i}\right )-(1-\theta_i\alpha _k\rho_{ki} )t\ln t-
\frac{\theta_i\alpha _k}{\beta_i} t^{1+\beta_i\rho_{ki}}
\end{equation}
and 
\begin{equation}\label{kfc3}K^I(t)=t-\ln t.\end{equation}
Notice that: (\ref{kfc0}) is the Kendall's function of the Marshall-Olkin copula 
with parameters $\alpha _i$ and $\alpha _k$; (\ref{kfc1}) is a Kendall's function of type (\ref{fclay0}) with parameters $\theta=\theta_k\alpha_i\rho_{ik}$ and 
$\beta=\beta_k\rho_{ik}$ (that represents the effect of the dependence between $Y_k$ and $X_k$ on the resulting dependence structure of $(T_i,T_k)$);
simmetrically, (\ref{kfc2}) is a Kendall's function of type (\ref{fclay0}) with parameters $\theta=\theta_i\alpha_k\rho_{ki}$ and 
$\beta=\beta_i\rho_{ki}$; (\ref{kfc3}) is the Kendall's function of the independence copula. As a consequence, we get a very meaningful decomposition of the Kendall's tau:
$$\tau_{ik}=\tau^{MO}_{ik}+\bar\tau^{(i)}_{0,k}+\bar \tau _{0,i}^{(k)}
$$
where 
$$\bar\tau^{(i)}_{0,k}=\alpha _i\rho_{ik}\theta_k\frac{\rho_{i,k}\beta_k}{\rho_{i,k}\beta_k+2}\text{ and }
\bar\tau^{(k)}_{0,i}=\alpha _k\rho_{ki}\theta_i\frac{\rho_{k,i}\beta_i}{\rho_{k,i}\beta_i+2}$$
are Kendall's tau of type (\ref{clay0}) with suitably modified parameters.

It follows that
\begin{equation}\label{clayy}\begin{aligned}\tau_{T_k,X_0}&=\alpha _k+(1-\alpha _k)\theta_k\frac{(1-\alpha_k)\beta_k}{(1-\alpha _k)\beta _k+2}=\\
                              &=\tau^{MO}_{T_k,X_0}+\tau_{0,k}^*
                             \end{aligned}\end{equation}
where $\tau^{MO}_{T_k,X_0}$ is the Kendall's tau between the observed lifetime and the systemic shock arrival time in the Marshall-Olkin model and 
$\tau_{0,k}^*$ is a Kendall's tau of type (\ref{clay0}) with parameters rescaled by the coeffcient $1-\alpha _k$. 
\item Gumbel case ($\phi_j(x)=e^{-x^{\frac 1{\beta _j}}}$, $\beta _j\geq 1$, $j=i,k$).
If
$$\mathcal I(a,b,\beta)=\int_a^b\frac {1}{z^{\beta}(z+1)}dz,$$

$$\begin{aligned}
K_{i,k}(t)&=t-t\ln t\left \{1-\tau_{ik}^{MO}+\tau^{MO}_{ik}\left [\theta _k\left (1-\left (\frac{\theta _k\alpha _k}{1-\alpha _k(1-\theta _k)}\right )^{\beta _k-1}\right )+\right .\right .\\
&\left .\left .-\theta _i\left (1-\left (\frac{\theta _i\alpha _i}{1-\alpha _i(1-\theta _i)}\right )^{\beta _i-1}\right )\right ]+\right .\\
&\left .- (\beta_k-1)\left (\frac{\alpha_i\theta_k}{1-\alpha _i\theta_k}\right )^{\beta _k}\mathcal I\left (\frac{\alpha_i\theta_k}{1-\alpha _i\theta_k},\frac{\alpha_i\theta_k}{\tau_{ik}^{MO}\theta _k(1-\alpha _i\theta_k)}-1,\beta _k\right)+\right .\\
&\left .- (\beta_i-1)\left (\frac{\alpha_k\theta_i}{1-\alpha _k\theta_i}\right )^{\beta _i}\mathcal I\left (\frac{\alpha_k\theta_i}{1-\alpha _k\theta_i},\frac{\alpha_k\theta_i}{\tau_{ik}^{MO}\theta _i(1-\alpha _k\theta_i)}-1,\beta _i\right)
\right \}
\end{aligned}
$$
and
$$\begin{aligned}
\tau_{i,k}&=\tau_{ik}^{MO}-\tau^{MO}_{ik}\left [\theta _k\left (1-\left (\frac{\theta _k\alpha _k}{1-\alpha _k(1-\theta _k)}\right )^{\beta _k-1}\right )-\theta _i\left (1-\left (\frac{\theta _i\alpha _i}{1-\alpha _i(1-\theta _i)}\right )^{\beta _i-1}\right )\right ]+\\
&+ (\beta_k-1)\left (\frac{\alpha_i\theta_k}{1-\alpha _i\theta_k}\right )^{\beta _k}\mathcal I\left (\frac{\alpha_i\theta_k}{1-\alpha _i\theta_k},\frac{\alpha_i\theta_k}{\tau_{ik}^{MO}\theta _k(1-\alpha _i\theta_k)}-1,\beta _k\right)+\\
&+ (\beta_i-1)\left (\frac{\alpha_k\theta_i}{1-\alpha _k\theta_i}\right )^{\beta _i}\mathcal I\left (\frac{\alpha_k\theta_i}{1-\alpha _k\theta_i},\frac{\alpha_k\theta_i}{\tau_{ik}^{MO}\theta _i(1-\alpha _k\theta_i)}-1,\beta _i\right )\end{aligned}$$
Even if, as in the Clayton case, we can recognize that the resulting dependence is the sum of the Marshall-Olkin one and two different contributions arising
from the assumed dependence between $Y_i$ and $X_i$ and between $Y_k$ and $X_k$, unlike that case, the latter ones 
cannot be written in a closed form as modifications of the corresponding ones
in Example \ref{ex0}.

Moreover, we have
$$\begin{aligned}\tau_{T_k,X_0}&=\alpha _k-\alpha _k\theta _k\left (1-\left (\frac{\theta _k\alpha _k}{1-\alpha _k(1-\theta _k)}\right )^{\beta _k-1}\right )+\\
&+ (\beta_k-1)\left (\frac{\theta_k}{1-\theta_k}\right )^{\beta _k}\mathcal I\left (\frac{\theta_k}{1-\theta_k},\frac{1-\alpha_k(1-\theta_k)}{\alpha_k(1-\theta_k)},\beta _k\right)
\end{aligned}$$
 
 \end{enumerate}

\end{example}

\section{Application to the analysis of the systemic riskiness in the European banking system}\label{sifi}
In this section we apply the model presented and discussed in previous sections to the analysis of the riskiness of the so called too-big-to-fail banks 
in the European banking system. We define systemic riskiness as the capability of a bank to induce a systemic crisis (collapse) in the banking system:
in our model, this can be measured thought the degree of dependence between the idiosyncratic component of the risk of default of the bank and the
systemic shock arrival time that causes the simultaneous default of all the banks in the system, that is through the Kendall's tau $\tau_{X_0,X_j}$ of the vector $(X_0,X_j)$.
\medskip

For the empirical analysis we restrict to the setup considered in Examples \ref{ex1} and \ref{ex2}.
More specifically, we consider the case in which all bivariate underlying copulas, modeling the dependence structure between each idiosyncratic component
and the associated systemic shock component, are of Clayton type. Since Clayton copula exibits lower-tail dependence, we are assuming 
stronger dependence between each idiosyncratic shock arrival time $X_j$ and the corresponding systemic component arrival time $Y_j$ 
when they have a very high probability to occur: this is in line with the well known fact that dependence tends to increase in crisis periods. This choice has also the advantage to let us deal with very nice and meaningful formulas.
\medskip

In the assumed setup, the bivariate Kendall's tau of the pairs $(T_i,T_k)$ depend
on the the set of parameters 
$$\boldsymbol{\Theta}=(\alpha _1,\ldots ,\alpha _d,\theta_0,\ldots, \theta _d,\beta _1,\ldots ,\beta _d):$$
$\alpha _j$ represents the ratio between the systemic shock intensity and the marginal one; $\theta _j$ measures the contribution of each bank 
to the systemic shock intensity while $\theta _0$ measures the contribution of some completely independent shock;
the parameters $\beta _j$ are the parameters of the involved bivariate copulas. As shown in (\ref{cc1}), these parameters fully characterize the dependence
structure of the vector of observed lifetimes ${\bf T}$. 

The estimation technique will consist in a moment based approach, through which 
theoretical bivariate Kendall's tau will be fitted to empirical ones. 
Once the parameters are estimated, we can use (\ref{clay0}) to estimate the systemic riskiness of each bank.
\subsection{Data set}
Our data set consists of daily 5 years CDS quotes, from 01/01/2009 to 31/12/2016
of the European banks classifies as SIFI by the Financial Stability Board \footnote{See the report
``2016 list of global systemically important banks (G-SIBs)'' published by the Financial Stability Board, http://www.fsb.org/wp-content/uploads/2016-list-of-global-systemically-important-banks-G-SIBs.pdf}.
Data were downloaded from Datastream.
\medskip

We assume that all arrival times are exponentially distributed, that is, in the notation of Examples \ref{ex1} and \ref{ex2}, $G(x)=e^{-x}$: as a consequence, 
also observable lifetimes are exponentially distributed, with intensities $\lambda_0+\lambda _j$ for $j=1,\ldots ,d$.
We also assume a constant interest rate and costant Loss-Given-Default. Thanks to these assumptions, survival probabilities and intensities can be easily extracted from $CDS$ spreads (see Brigo and Mercurio, p.735-6).
\smallskip

Since a sample of default times is not available, we are not in the position to recover the empirical Kendall's tau from default times data.
However, the Kendall's tau is the difference between the proportion of concordant and discordant pairs of observations and an increase in the intensity of default 
corresponds to the perception that the default 
time is going to occur earlier: 
in the absence of more appropriate data, we recover intensities from the CDS spreads dataset and we assume as empirical Kendall's tau those 
estimated from intensities. As a consequence our analysis will be based on the information implied by the CDS.
\subsection{Estimation procedure}
Let $\hat \tau_{ik}$, $i=1,\ldots ,d-1$, $k=i+1,\ldots ,d$ be the estimated pairwise empirical Kendall's tau and $\tau_{ik}(\alpha _i,\alpha _k,\theta _i,\theta _k,\beta _i,\beta_k)$
be the corresponding theoretical ones. Parameters are estimated by solving
\begin{equation}\label{optimal}\boldsymbol{\hat\Theta} = \underset{\boldsymbol{\Theta}}{\operatorname{argmin}} \sum_{i=1}^{d-1}\sum_{k=i+1}^{d} \left(\hat
\tau_{i,k}-\tau_{i,j}(\alpha _i,\alpha _k,\theta _i,\theta _k,\beta _i,\beta_k)\right )^2. \end{equation}
This moment based procedure is a generalization of the Kendall's tau-based estimation procedure considered in Genest and Rivest (1993) to the 
multidimensional framework and it has been analyzed and studied in Mazo et al. (2015). 

\noindent The optimization required in (\ref{optimal}) is not a trivial task and can only be solved numerically.

\subsection{Results}
The procedure applied to all SIFI European banks does not provide a good fit and the global minimum remains far from 0. Things work much better
if one restrict the analysis to the banks that in the list of globally systemically important banks provided by the Financial Stability Board, are identified as
particularly systemically risky since they have associated buckets higher than $1$ (higher buckets correspond to higher 
level of systemic importance): 
BNP Paribas, Deutsche Bank, HSBC, Barclays.
\medskip

The estimation is conducted on a yearly basis and, once the parameters have been estimated, the Kendall's tau $\tau_{X_0,X_j}$ are evaluated according to 
(\ref{clay0}).
In Table \ref{tab1} we show the obtained values of the Kendall's tau between each idyosincratic component $X_j$ and 
the systemic shock $X_0$.

\begin{center}
 \begin{table}
\caption{\footnotesize Yearly Kendall's tau values $\tau_{X_0,X_j}$.}\label{tab1}\vskip 1pc
\begin{center}
\begin{tabular}{|c|c|c|c|c|}
\hline
&DEUTSCHE BANK &BNP PARIBAS&BARCLAYS&HSBC\\
\hline
2009&0.66425573&0.04839780&0.25905985&0.01702301\\
\hline
2010&0.1961576&0.0000000&0.2135096&0.5891973\\
\hline
2011&0.00000000&0.17627564&0.04529224&0.77736158\\
\hline
2012&0.00000000&0.02586159&0.85703128&0.11605882\\
\hline
2013&0.0000000&0.3797177&0.0000000&0.6192373\\
\hline
2014&0.07860066&0.00000000&0.00000000&0.92038204\\
\hline
2015&0.0000000&0.3477391&0.0000000&0.6512462\\
\hline
2016&0.1911062&0.0000000&0.0000000&0.8068463\\
\hline
\end{tabular}
\end{center}
\end{table}
\end{center}

It worths to mention that the global minimum in (\ref{optimal}) is very close to $0$ in years 2009-2011. In particular, the fit is particularly good in year 2009: in this year the US banking crises has spread in Europe with its systemically relevant effects and,
as sown in Table \ref{tab1}, all banks
are systemically risky, in the sense considered in this paper, even if with different degrees. Table \ref{tab1} shows that, even if the capability of each bank to cause
the bankrupcy of the whole banking system changes with time, HCBS is globally the most risky in the analyzed period. Comparing the obtained results  with the 
available 2015 and 2016 reports
\footnote{see ``2015 list of global systemically important banks (G-SIBs)'', http://www.fsb.org/wp-content/uploads/2015-update-of-list-of-global-systemically-important-banks-G-SIBs.pdf} of
the Financial Stability Board (based, respectively, on end 2014 and end 2015 data), we observe the the extraordinary high degree of systemic riskiness estimated for HSBC in 2014 ($92\%$)
is in line with the association of this bank to bucket 4 (the highest) in the 2015 report, while its reduced degree of riskiness estimated in 2015 ($65\%$) is in line with the doungrade of HSBC 
to bucket 3 in the 2016 report. Additionally, according to Table \ref{tab1} Barclays can be classified as the less risky (in the sense considered in this paper) in recent past years: this is consistent with
the fact that, among the considered banks, it is the only one to which it is assigned bucket 2 in the 2016 report.
\begin{center}
 \begin{table}
\caption{\footnotesize Yearly Kendall's tau values $\tau_{X_0,T_j}$.}\label{tab2}\vskip 1pc
\begin{center}
\begin{tabular}{|c|c|c|c|c|}
\hline
&DEUTSCHE BANK &BNP PARIBAS&BARCLAYS&HSBC\\
\hline
2009&0.6939868&0.6954599&0.9342763&0.9406805\\
\hline
2010&0.8664811&0.8623019&0.8742886&0.7039019\\
\hline
2011&0.9176120&0.9161709&0.8260221&0.8066270\\
\hline
2012&0.8962147&0.8163799&0.8570315&0.7563900\\
\hline
2013&0.8559030&0.8049459&0.8226217&0.8344751\\
\hline
2014&0.6916339&0.7916272&0.9123942&0.9203822\\
\hline
2015&0.8375908&0.7653512&0.7463275&0.8434409\\
\hline
2016&0.2536626&0.5585990&0.6998114&0.8206532\\
\hline
\end{tabular}
\end{center}
\end{table}
\end{center}
In Table \ref{tab2} we list the Kendall's tau values between each observed lifetime $T_j$ and the systemic shock arrival time $X_0$.
We notice that in some cases the values $\tau_{X_0,X_j}$ and $\tau_{X_0,T_j}$ are very close each other: this is the case of Deutsche Bank in 2009,
Barclays in 2012 and HSBC in 2011, 2014 and 2016. As noticed at the end of Example \ref{ex1}, this is due to the fact that the dependence of the bank lifetime 
with the systemic risk
is essentially given by its capability to induce a systemic shock and in a negligible way by the fact that it is subjected to the systemic shock itself: this is a clear evidence of riskiness.  

\section{Conclusions}\label{conclusion}
In this paper we have introduced a generalization of the Marshall-Olkin distribution in which some non-exchangeable dependence among the underlying shocks arrival times is assumed.
More specifically, we have assumed that each lifetime is the first arrival time between an idiosyncratic and a systemic shock and, unlike the standard Marshall-Olkin model,
we have assumed some dependence between each idiosyncratic arrival time and the systemic one: the resulting model is particularly suitable to model
situations in which lifetimes influence each other only through the systemic shock arrival time on which they are dependent.
The obtained joint distribution of lifetimes is investigated: its singularity analyzed and its dependence properties studied through the induced copula functions and
the associated pairwise Kendall's function and Kendall's tau. The dependence structure is the composition of a Marshall-Olkin type dependence and the 
assumed dependence of each idiosyncratic component with the systemic shock arrival time: the higher the second component,
the more risky is the considered entity.
\medskip

The model is applied to the analysis of the systemic riskiness of SIFI type European banks.
Results show that the model gives a better fits if one restrict to particularly ``big'' SIFI banks, according to 
the Financial Stability Board buckets classification: BNP Paribas, Deutsche Bank, HSBC, Barclays.
\medskip

The obtained results allow to classify the systemic riskiness of these banks according to their capability to induce the 
simultaneous default of all
the system. This is an information that could be used, in addition to the already used ones, to more completely classify the riskiness of a bank.

\end{document}